\documentclass{article}

\PassOptionsToPackage{numbers}{natbib}
\usepackage[final]{neurips_2021}

\usepackage[utf8]{inputenc} % allow utf-8 input
\usepackage[T1]{fontenc}    % use 8-bit T1 fonts
\usepackage{hyperref}       % hyperlinks
\usepackage{url}            % simple URL typesetting
\usepackage{booktabs}       % professional-quality tables
\usepackage{amsfonts}       % blackboard math symbols
\usepackage{nicefrac}       % compact symbols for 1/2, etc.
\usepackage{microtype}      % microtypography
\usepackage{xcolor}         % colors

\usepackage{algorithm}
\usepackage[noend]{algpseudocode}
\usepackage{graphicx}
\usepackage{float, subcaption}
\usepackage{amsmath,amssymb,amsthm, bm, mathabx}
	\usepackage{bbm} % for bold type 1 for indicator function
	\let\emptyset\varnothing % nice-looking empty set
\usepackage{booktabs} % better looking tables
\usepackage{mathtools}
\usepackage{bbm}
\usepackage{tikz}
\usepackage{tikz-3dplot}
\usepackage{pgfplots}
\usepgfplotslibrary{ternary}
\usetikzlibrary{calc,shapes}
\usepackage{dsfont}
\usepackage{enumitem}
\usepackage{amsthm}
\usepackage{thm-restate}
\usepackage{multirow}
\usepackage{newfloat}
\usepackage{wrapfig}
\usepackage{mleftright}
\usepackage[capitalise,noabbrev]{cleveref}
\usepackage{caption}

\usepackage{pifont}
\usepackage[normalem]{ulem}
\title{Public Information Representation\\for Adversarial Team Games}

\author{Luca Carminati, Federico Cacciamani, Marco Ciccone, Nicola Gatti\\
Dipartimento di Elettronica, Informazione e Bioingegneria, Politecnico di Milano\\
Piazza Leonardo da Vinci, 32, 20133, Milano, Italy\\
luca5.carminati@mail.polimi.it, \{federico.cacciamani, marco.ciccone, nicola.gatti\}@polimi.it\\
}

\algrenewcommand\algorithmicindent{1.0em}%

\begin{document}

\newtheorem{theorem}{Theorem}
\newtheorem{corollary}{Corollary}[theorem]
\newtheorem{lemma}[theorem]{Lemma} % actually separate lemmas from theorem numbering, in order to have \autoref work well
\newtheorem{definition}{Definition}
% argmin and argmax operators for better centering of arguments

\newcommand{\argmin}{arg\min}
\newcommand{\argmax}{arg\max}
\newcommand{\Reals}{\mathbb{R}}
\newcommand{\team}{\mathcal{T}}
\newcommand{\Z}{\mathcal{Z}}
\newcommand{\St}{\mathcal{S}}
\newcommand{\N}{\mathcal{N}}
\newcommand{\algoname}[1]{\textnormal{\textsc{#1}}}

\definecolor{mygreen}{rgb}{0.0, 0.5, 0.0}
\newcommand{\nb}[3]{{\colorbox{#2}{\bfseries\sffamily\scriptsize\textcolor{white}{#1}}}{\textcolor{#2}{\sf\small\textit{#3}}}}

\newcommand{\lu}[1]{\nb{Luca}{blue}{#1}}
\newcommand{\ma}[1]{\nb{Marco}{red}{#1}}
\newcommand{\fe}[1]{\nb{Federico}{mygreen}{#1}}
\newcommand{\nc}[1]{\nb{Nicola}{violet}{#1}}

\maketitle

\begin{abstract}

The study of sequential games in which a team plays against an adversary is receiving an increasing attention in the scientific literature.
Their peculiarity resides in the asymmetric information available to the team members during the play which makes the equilibrium computation problem hard even with zero-sum payoffs. The algorithms available in the literature work with implicit representations of the strategy space and mainly resort to \textit{Linear Programming} and \emph{column generation} techniques. Such representations prevent from the adoption of standard tools for the generation of abstractions that previously demonstrated to be crucial when solving huge two-player zero-sum games. 
Differently from those works, we investigate the problem of designing a suitable game representation over which abstraction algorithms can work. In particular, our algorithms convert a sequential team-game with adversaries to a classical \textit{two-player zero-sum} game. 
In this converted game, the team is transformed into a single coordinator player which only knows information common to the whole team and prescribes to the players an action for any possible private state. 
%We named this procedure \textit{Public Team Conversion}, and it returns an extensive-form game maintaining most of the structure of the original game. 
Our conversion enables the adoption of highly scalable techniques already available for two-player zero-sum games, including techniques for generating automated abstractions. 
Because of the \textsf{NP}-hard nature of the problem, the resulting Public Team game may be exponentially larger than the original one. To limit this explosion, we design three pruning techniques that dramatically reduce the size of the tree. 
% Finally, we present experimental results obtained by applying our techniques to standard benchmarks in the field (\textit{Kuhn} and \textit{Leduc Poker} games) and apply state-of-the-art algorithms for 2p0s games to the resulting game, showing the effectiveness of our approach.
Finally, we show the effectiveness of the proposed approach by presenting experimental results 
on \textit{Kuhn} and \textit{Leduc Poker} games, obtained by applying state-of-art algorithms for two players zero-sum games on the converted games.
\end{abstract}
\section{Introduction}\label{sec:intro}
% team problem importance
Research efforts on imperfect-information games customarily focus on two-player zero-sum (we refer to these games with ``2p0s'' from here on) scenarios, in which two agents act in the same environment, receiving opposite payoffs. In this setting, superhuman performances have been achieved even in huge game instances, such as \textit{Poker Hold'em} \cite{BrownS17,Brown885,Moravck2017DeepStackEA} and \textit{Starcraft II}  \cite{Vinyals2019GrandmasterLI}.
%
%On the other hand, autonomous agents are increasingly employed in many different contexts. This raises the issue of designing techniques to produce robust agents also in the case in which multiple agents act in the same environment, possibly sharing their rewards.
%
The customary approach for 2p0s games is based on the generation of a game abstraction that is used offline to find a blueprint strategy that is refined online during the play.
Instead, the problem of solving games with multiple players is far from being addressed and lacks well-assessed techniques.

In our work, we focus on \emph{adversarial team games} in which a team of two (or more) agents cooperates against a common adversary (interestingly, our result directly applies to the case in which even the adversary is a team of players). In particular, we focus on the \textit{ex-ante coordination} scenario in which the team members agree on a common strategy beforehand and commit to playing it during the game without communicating any further. Examples of such a scenario are collusion in poker games, the defenders in the card-playing stage of \textit{Bridge}, and a team of drones acting against an adversary. Technically speaking, the team members share the same payoffs and coordinate against an adversary having opposite payoffs, in face of private  information given separately to each team member.

% history of the field
A natural way to characterize the desired behavior by each agent is the  \textit{team-maxmin} equilibrium (TME) solution concept, defined by Von Stengel in \cite{VonStengelKoller97}. In \cite{Celli2018ComputationalRF}, this concept is extended to extensive-form games considering different means of communication. The team-maxmin equilibrium with correlation (TMECor) characterizes optimal rational behavior in the \emph{ex-ante} coordination scenario we consider in this work. While the existence and uniqueness of value of a TMECor are guaranteed, the computation of an equilibrium corresponding to this solution concept is proven to be an inapproximable, \textsf{NP}-hard problem \cite{Celli2018ComputationalRF}. Therefore, the main open challenge is to efficiently find a strategy profile corresponding to a TMECor in a generic adversarial team game.

%%% --- maybe move to a related works section with the public information stuff from Nayyar and Forester?
% current approaches
\textbf{Related works.} %Classical algorithms already developed for 2p0s games cannot be applied directly to team games since the private information held by each team member cannot be shared with the other team members; therefore the problem of finding a TMECor in an adversarial team game has been customarily addressed through the design of \textit{ad hoc} techniques. 
Two main classes of approaches are available in the literature to address adversarial team games.

The first class of approaches is based on mathematical programming. These approaches commonly resort to  \textit{column generation}. They constrain the team to play a probability distribution over a finite set of correlated plans and iteratively add new correlated plans giving the maximum increase in value for the team. % until no new plans can be added. 
% The constraint to draw plans from a finite set for the team 
This constraint allows a polynomial-time formulation of the TMECor problem, while the plan to add is determined by using an oracle working in an implicit formulation of the normal-form representation of the original game in which the private information of a team player can be ignored during the selection of a plan. \textit{Hybrid Column Generation} (HCG)~\cite{Celli2018ComputationalRF} iteratively solves two linear programs to determine the TMECor strategy profile of the team and the adversary considering the limited set of plans available, while uses an integer linear program to find the best response joint plan to add to the currently considered plans given the current strategy of the adversary. \textit{Fictitious Team Play} (FTP)~\cite{Farina2018ExAC} is instead a \textit{fictitious play}~\cite{BrownFP} procedure running on a modified representation of the original game; in this representation, one of the two team members selects a pure strategy at the start of the game, while the other team member plays against the adversary in the original game considering his teammate fixed as specified by their chosen plan. \textit{Faster Column Generation} (FCG)~\cite{Farina2021ConnectingOE} is a column-generation algorithm similar to HCG, but working in a more efficient semi-randomized correlated plans representation, prescribing a pure strategy for one team member and a mixed strategy for the other. This representation also allows a cost-minimizing formulation of the best response problem as a dual of the program to find a TMECor considering the restricted set of plans. Overall, those procedures are general and simple to implement, but the use of integer linear programs strongly limits the scalability of these exact formulations.

The second class of approaches is based on Multi-Agent Reinforcement learning (MARL). These approaches avoid considering plans over the entire game, and instead explicitly model the correlating coordination signal in the original extensive form adversarial team game. \textit{Soft Team Actor-Critic} (STAC)~\cite{Celli2019CoordinationIA} fixes a number of possible uniform signals at the start of the game and uses a modified actor-critic procedure to converge to a strategy for each player for each signal. \textit{Signal Mediated Strategies} (SIMS)~\cite{Cacciamani2021MultiAgentCI} works with a \textit{perfect-recall refinement} of the original game, populates a buffer of trajectories sampled from the optimal strategy for the joint team against the adversary in this refined game, and learns a distributed strategy for each team member in the original game from the trajectories in this buffer. Overall, RL techniques are more scalable than the mathematical programming ones, but present convergence issues. Indeed, STAC does not have any guarantees of converging to equilibrium, while SIMS is guaranteed to converge only on a particular class of Adversarial Team Games.
%%% --- 

The above approaches work with an implicit strategy space representation that does not allow the adoption of the standard tools previously developed for 2p0s games, including abstraction generation, no-regret game solving algorithms, and subgame solving.

% my contributions
\textbf{Original contributions.} We propose an algorithmic procedure, called \algoname{PublicTeamConversion}, to convert a generic instance of a team game into a 2p0s game, and provide three algorithms, each returning an information-lossless abstraction of the converted game. Our conversion enables the adoption of customary techniques for 2p0s abstraction generation, no-regret game solving, and subgame solving.
Surprisingly, we prove that any state/action abstraction applied to the extensive-form game can be captured by our representation, while the reverse does not hold, thus showing that our game representation is more expressive than the extensive-form game.
Furthermore, we formally prove that a Nash Equilibrium in the converted game corresponds to a TMECor in the original game. In addition, we present vEFG, an alternative description of extensive-form games better suited for the characterization of public information.
Our information-lossless abstractions allow a dramatic reduction of the size of the game tree, producing an abstract game that, in most of our experiments, has a size that is the square root of the size of the original game. Most importantly, our algorithms directly produce the abstract version of the game tree without the need for the original (non-abstracted) game.

% The core idea of our approach is to combine the sharing of all team members' strategies with the use of public information to determine a belief over the possible private information of all team members.  As the game progresses, the information state of each player is enriched by the notion that other team members may or may not be in a specific private state given the deterministic strategy shared at the start of the game and the observed public actions played. In fact, playing a specific action when the overall strategy is known, communicates part of the private information.  Following this intuition, we propose an explicit representation for the game by joining the team members into a single coordinator who has an information state based only on the public information shared among team members and that prescribes a deterministic action for each possible private state of the player.
% \textcolor{red}{SCRIVERE FORNIAMO PROVE ALTERNATIVE CHE IN ALCUNI GIOCHI IL PROBLEMA E' POLINOMIALE IN QUANTO L'ALBERO PRODOTTO E PRUNATO E' POLINOMIALE} \textcolor{blue}{Dove lo facciamo?}

% related works
Our work builds upon Nayyar et al. in \cite{Nayyar2013DecentralizedSC}; similarly to our work, a shared coordinator, living in the public information state of the team, prescribes an action to the team members for any possible private state allowing the training of a common strategy for the team.  This coordinator-based transformation is recently employed in the fully cooperative setting of \textit{Hanabi} \cite{Foerster2019BayesianAD,Sokota2021SolvingCG}. However, their transformation applies to the context of decentralized stochastic control in a fully cooperative scenario, whereas we generalize such coordinator transformation in an adversarial team game setting.

\section{Preliminaries}\label{sec:prelim}
In this section, we introduce the basic concepts and definitions that we need throughout this work (for more details, we refer the interested reader to \cite{shoham-leyton}).
\paragraph{Extensive-Form Games and Adversarial Team Games}

The basic model for sequential interactions among a set $\mathcal{N}$ of multiple agents with private information is the \emph{Extensive-Form Game with imperfect information} (EFG). 
If chance is present in the game, we enrich the set of players with the chance player $c$, and chance probabilities are denoted as $\sigma_c$. 
An EFG defines a tree where the set of nodes is denoted by $\mathcal{H}$ and the leaves (terminal nodes of the game) are denoted by $\mathcal{Z}\subseteq\mathcal{H}$. 
The player acting at a node $h\in\mathcal{H}$ is identified by the function $\mathcal{P}(h) \in \mathcal{N}$, and the set of actions that she can choose at $h$ are $A(h)$.
Let $u_p: \mathcal{Z}\to\Reals$ be the payoff function of player $p$ that maps every terminal node to a utility value.
In order to account for imperfect information, we use \emph{information sets} (infosets).
An infoset $I\subseteq\mathcal{H}$ is a partition of the nodes of the tree in which a player acts such that they are indistinguishable to her. 
We denote the set of player $p$ infosets as $\mathcal{I}_p$.
With a slight abuse of notation, we denote the player acting at infoset $I$ as $\mathcal{P}(I)$ and her available actions at that infoset as $A(I)$.
Furthermore, the action space is defined as $\mathcal{A} = \bigcup_{p\in\mathcal{N}}\mathcal{A}_p$, where $\mathcal{A}_p = \bigcup_{I\in\mathcal{I}_p} A(I)$.

In this work, we are interested in Adversarial Team Games (ATGs).
An ATG is a $N$-player EFG in which a \emph{team} $\mathcal{T}\subseteq\mathcal{N}$ of players plays against a single opponent ($\mathcal{N} = \mathcal{T}\cup\left\{ o\right\}$, where $o$ is the opponent). 
Here, the team is represented by a set of agents that share the same utility function.
Formally $\forall p \in \mathcal{T}$, $u_p = u_\mathcal{T}$ for some function $u_\mathcal{T}$. 
We restrict our analysis to zero-sum ATG, \emph{i.e.}, games in which $u_{\mathcal{T}} = -u_o$.
For an EFG, a deterministic timing is a labeling of the nodes in $\mathcal{H}$ with non-negative real
numbers such that the label of any node is at least one higher than the label of its parent. 
A deterministic timing is exact if any two nodes in the same information set have the same label. 
An EFG is \emph{1-timeable} if it admits a deterministic exact timing such that each node's label is exactly one higher than its parent's label.
Furthermore, we require that all the players are with \emph{perfect recall}, \emph{i.e.,} they never forget information they previously knew during the game.

\paragraph{Strategies and Nash Equilibrium}
There are several possible ways of representing strategies. 
A \emph{behavioral strategy} $\sigma_p: \mathcal{I}_p\to\Delta^{|A(I)|}$ is a function that maps each infoset to a probability distribution over available actions. 
A \emph{normal-form plan} (or \emph{pure strategy}) $\pi_p\in\Pi_p:= \bigtimes_{I\in\mathcal{I}_p}A(I)$ is a tuple specifying one action for each infoset, while a \emph{normal-form strategy} $\mu_p\in\Delta^{|\Pi_p|}$ is a probability distribution over normal-form plans.
A \emph{reduced normal-form strategy} $\mu^\star_p$ is obtained from a normal-form strategy $\mu_p$ by aggregating plans distinguished by actions played in unreachable nodes.
With a slight abuse of notation, $\forall p \in\mathcal{N}$ we write with $\sigma_p[z]$ (respectively $\mu_p[z]$) the probability of reaching terminal node $z\in\Z$ when following strategy $\sigma_p$ (resp. $\mu_p$). 
A strategy profile is a tuple associating a strategy to each player in the game. 
We denote normal-form strategy profiles as $\boldsymbol{\mu}$ and behavioral strategy profiles as $\boldsymbol{\sigma}$. 
Given a strategy profile $\boldsymbol{\mu}$, we denote with $\mu_p$ the strategy of player $p\in\N$ and with $\boldsymbol{\mu}_{-p}$ the strategies of all the other players. 
With an abuse of notation, the expected utility for player $p$ when she plays strategy $\mu_p$ and all the other players play strategy $\boldsymbol{\mu}_{-p}$ is $u_p(\mu_p, \boldsymbol{\mu}_{-p})$.
Furthermore, we define the \emph{best response} of player $p$ to strategy profile $\boldsymbol{\mu}_{-p}$ as the strategy that maximizes player $p$'s utility against strategy $\boldsymbol{\mu}_{-p}$.
Formally, $BR_p(\boldsymbol{\mu}_{-p}) := \argmax_{\mu} u_p(\mu, \boldsymbol{\mu}_{-p})$.
A strategy profile $\boldsymbol{\mu}$ is a Nash Equilibrium (NE) if it is stable to respect to unilateral deviations of a single player. 
Formally, $\boldsymbol{\mu}$ is a NE if and only if $\forall p \in \mathcal{N}$, $\mu_p \in BR_p(\boldsymbol{\mu}_p)$.

% \fe{Cambiare la notazione ? ($\sigma$ per la forma normale, $\pi$ per la behavioral)}
%
% \input{content/2_preliminaries}
%
%\input{content/3_tmecor}
\section{How to Enforce Coordination Between Team Members}\label{sec:tmecor}

When considering ATGs, the problem of computing optimal strategies for the team of agents becomes inherently much more complex than with 2p0s games. 
This is basically due to team players' need for coordinating their strategies to maximize their utility. 
In practice, such a requirement on coordination is translated into the formulation of an \emph{ad hoc} solution concept for ATGs. 
Von Stengel and Koller, in \cite{VonStengelKoller97}, introduce the \emph{Team Maxmin Equilibrium} (TME), which is defined as the NE in which the team's utility is maximized. 
The solution concept described by the TME exhibits some appealing characteristics: (i) in some cases it can be arbitrarily more efficient in terms of payoffs for the team with respect to a NE, (ii) it is unique (except for degeneracies), thus it does not suffer from equilibrium selection issues \cite{Basilico2016teammaxmin,Basilico2017ComputingTT,Celli2018ComputationalRF}.
However, it is not easy to be computed as it does not admit a convex mathematical formulation \cite{Farina2021ConnectingOE}.

If the coordination of team's strategies is done \emph{ex-ante}\footnote{Note that different forms of coordination are possible, (\emph{e.g.}, intra-play) but they require communication capabilities between the team members that may not be allowed in arbitrary games.} (\emph{i.e.}, before the beginning of each game), a new solution concept can be introduced, called \emph{Team Maxmin Equilibrium with a Correlation device} (TMEcor) \cite{Celli2018ComputationalRF}. 
Such coordination has the advantage to increase the expected utility of the team \cite{Celli2018ComputationalRF}. Moreover, differently from what happens with the TME, the TMEcor can be computed through a Linear Program working with joint normal-form plans of the team players: 
\begin{equation}\label{eq:tmecor}
\begin{array}{l}\displaystyle
\max_{\mu_\team}\min_{\mu_o} \sum_{z\in\Z}\hspace{.4cm} \mu_\team[z]\,\mu_o[z]\,u_\team(z)\\
\hspace{.3cm}\textnormal{s.t. }\hspace{.5cm}\mu_\team\in\Delta(\bigtimes\limits_{p\in\team}\Pi_p)\\[4mm]
\hspace{1.3cm}\mu_o\in\Delta(\Pi_o).
\end{array}
\end{equation}

In general, the size of the strategy space of the team $\bigtimes_{p\in\team} \Pi_p$ grows exponentially in the number of team members, thus increasing the complexity of solving the LP in Equation~\eqref{eq:tmecor}.
There are, however, some specific cases in which we can compute or approximate efficiently the TMEcor, simply by considering all the team members as a unique \emph{meta-player} $\team$ and applying state-of-the-art algorithms for finding NE in 2p0s games.
In order to use these techniques, we need that the meta-player has perfect recall. 
The requirement of perfect recallness for the team is satisfied only when the information available to the team players is \emph{symmetric}, (\emph{i.e.}, when they all observe an action or none of them does).
Unfortunately, this is not usually the case. 
In principle, the sources of imperfect recallness for the team can be several\footnote{We refer the interested reader to Appendix~\ref{app:info_structure} for a more detailed description of the aforementioned cases.
}: (i) non-visibility over a team member's actions, (ii) non-visible game structure, (iii) private information disclosed only to a subset of team members.
Mathematical programming algorithms \cite{Celli2018ComputationalRF,Farina2018ExAC,Farina2021ConnectingOE}  address the problem of imperfect recallness for the team by considering the space of joint strategies. When the support of the TMEcor strategy for the team is small, they can efficiently converge to the equilibrium.
On the other hand, MARL solutions, like SIMS \cite{Cacciamani2021MultiAgentCI}, are capable of solving case (i), but fail to deal with the other sources of imperfect recallness.
In this work, exploiting the concept of \emph{public information}, we introduce a method allowing us to deal with all the three sources of imperfect recallness, thus opening to the possibility of using state-of-the-art RL algorithms (\emph{e.g.}, Deep-CFR\cite{Brown2019DeepCR}) to solve ATGs.
We do this first by defining a game formalism that enriches EFGs with information on public observability of actions, and then by formulating a method that casts the original ATG into an equivalent 2p0s EFG. 
\section{vEFG Representation}\label{sec:vefg}
In this section, we introduce the concept of \emph{Extensive Form Game with visibility} (vEFG), a game representation that overcomes the limits of EFGs, allowing the characterization of the information common to a set of players.
The vEFG is an EFG enriched with a visibility function $Pub_p(a): \mathcal{A}\to\{\text{seen},\text{unseen}\}$ that specifies whether an action $a\in\mathcal{A}$ performed by any player has been observed by player $p$.
With abuse of notation, we define the $Pub$ function also for a set of players $\mathcal{P}\subseteq\mathcal{N}$.
Formally, $Pub_{\mathcal{P}}: \mathcal{A}\to\{\text{pub},\text{priv},\text{hidden}\}$ such that:
\begin{align*}
& Pub_{\mathcal{P}}(a) = \mathrm{pub} \iff \forall p \in \mathcal{P}:Pub_p(a)=\mathrm{seen}\\
& Pub_{\mathcal{P}}(a) = \mathrm{hidden} \iff \forall p \in \mathcal{P}:Pub_p(a) = \mathrm{unseen}\\
& Pub_{\mathcal{P}}(a) = \mathrm{priv} \text{ otherwise}.
\end{align*}

The game tree $\St$ associated with the vEFG is called \emph{public tree}.
A state $S\in\St$ is called \emph{public state}.
Two game histories belong to the same public state if they share the same public actions and differ only by their private actions. 
Formally, $h,h'$ belong to the same public state if and only if ${(a)_{a\in h : Pub_{\mathcal N}(a)=\mathrm{pub}} = (a')_{a'\in h' : Pub_{\mathcal N}(a')=\mathrm{pub}}}$.
We also use $\St(h)$ to indicate the public state associated with a history $h$.
Moreover, given a subset of players $X \subset \mathcal N$, we can define the public tree for $X$ as $\St_X$ such that it holds the following property ${(a)_{a\in h : Pub_{X}(a)=\mathrm{pub}} = (a')_{a'\in h' : Pub_{X}(a')=\mathrm{pub}}}$.

The information set structure $\mathcal I = (\mathcal I_p)_{p\in\mathcal N}$ can be recovered from $Pub_p$, by considering in the same infoset the histories corresponding to the same sequence of actions observed by a player.
Formally, $\forall h,h' \in \mathcal H : h,h' \in I \subset \mathcal I_p$ if and only if $\mathcal P(h) = \mathcal P(h') = p$ and $(a)_{a\in h : Pub_p(a)=\text{seen}} = (a')_{a'\in h' : Pub_p(a')=\text{seen}}$.
Note that in vEFGs we have no notion of forgetting actions, thus imperfect recall situations in which an observation is forgotten by a player cannot be represented.

As the only assumption for our method, we restrict the class of games we consider, introducing the concept of \textit{public turn-taking}.
This concept characterizes games in which the sequence of acting players is common knowledge across all players. 
Intuitively, each player knows when other players played an action in the past, even if the game has imperfect information and the specific action played may be hidden. 
This is a refinement to the concept of 1-timeability, in which not only the length of the history is identical but also the sequence of players.
\begin{definition}[Public turn-taking property]
A vEFG is public turn-taking if:
\begin{equation*}
\forall I \in \mathcal I, \forall h,h' \in I : (\mathcal P(g))_{g \sqsubseteq h} = (\mathcal P(g'))_{g' \sqsubseteq h'}
\end{equation*}
\end{definition}
\begin{restatable}[Transformation into public turn-taking game]{theorem}{sizepublicturn}
Any vEFG can be made public turn-taking by adding player nodes with a single noop action. The size of the resulting game tree is $(|\mathcal N| + 1)|\mathcal H|^2$.
\label{le:turnTaking}
\end{restatable}
The proof is in Appendix~\ref{app:proofs}.
\section{Public Team Conversion Algorithm}\label{sec:conversion}

%By combining the various solutions proposed for the different types of private information, we can formulate a conversion procedure able to convert an adversarial team game into a 2p0s game, in which a coordinator player takes the strategic decision on behalf of the team.
We present the algorithmic procedure to convert an ATG into a 2p0s game, in which a coordinator player takes the strategic decision on behalf of the team. We assume that the ATG is a vEFG with the public-turn-taking property. The pseudo-code is shown in Algorithm~\ref{alg:PubTeam}.
\begin{wrapfigure}[37]{R}{.50\textwidth}
\begin{minipage}{.5\textwidth}
\begin{algorithm}[H]%[!htb]
% \algsetup{linenosize=\tiny}
\caption{Public Team Conversion}\label{alg:PubTeam}
\begin{algorithmic}[1]
	\Function {ConvertGame}{$\mathcal G$}
		\State initialize $\mathcal G'$ new game 
		\State $\mathcal N' \gets \{t,o\}$ 
		\State $h'_{\emptyset} \gets$  \Call{PubTeamConv}{$h_{\emptyset}, \mathcal G, \mathcal G'$} %\Comment{new game root}
		\State \Return{$\mathcal G'$}
	\EndFunction
	\item[] % to skip line
	\Function{PubTeamConv}{$h$, $\mathcal G$, $\mathcal G'$}
		\State initialize $h' \in \mathcal H'$
  		\If {$h \in \mathcal Z$} %\Comment{terminal node}
  			\State $h' \gets h' \in \mathcal Z'$
  			\State $u'_p(h') \gets u_p(h) \quad \forall p \in \mathcal N$
  		\ElsIf {$\mathcal P(h) \in \{o, c\}$} %\Comment{opponent or chance}
  			\State $\mathcal P'(h') \gets \mathcal P(h)$
  			\State $\mathcal A'(h') \gets \mathcal A(h)$
  			\If {$h$ is chance node}
  				\State $\sigma_c'(h') = \sigma_c(h)$
  			\EndIf
  			\For {$a' \in \mathcal A'(h')$}
  				\State $Pub_t'(a') \gets$ check $Pub_{\mathcal T}(a') = \mathrm{pub}$
  				\State $Pub_o'(a') \gets Pub_o(a')$
  				\State $h'a' \gets$ \Call{PubTeamConv}{$ha'$, $\mathcal G$, $\mathcal G'$}
  			\EndFor 
  		\Else %\Comment{team member}
  			\State $\mathcal P'(h') = t$
  			\State $\mathcal A'(h') \gets \bigtimes_{I \in \mathcal  S_{\mathcal T}(h)} \mathcal A(I)$ \label{lst:line:prescription} %\Comment{prescriptions}
  			\For {$\Gamma' \in \mathcal A'(h')$}
  				\State $Pub_t'(\Gamma') \gets \mathrm{seen}$
  				\State $Pub_o'(\Gamma') \gets \mathrm{unseen}$
  				\State $a' \gets \Gamma'[\mathcal I(h)]$ %\Comment{extract chosen action}
  				\State initialize $h'' \in \mathcal H'$
  				\State $\mathcal A'(h'') \gets \{a'\}$
 				\State $\mathcal P(h'') = c$
  				\State $Pub_t'(a') \gets \mathrm{seen}$
 				\State $Pub_o'(a') = Pub_o(a')$
 				\State $\sigma_c'(h'') =$ play $a'$ with probability 1
  				\State $h''a' \gets$ \Call{PubTeamConv}{$ha'$, $\mathcal G$, $\mathcal G'$}
  				\State $h'\Gamma \gets h''$ 
  			\EndFor 
  		\EndIf
  		\State \Return $h'$
	\EndFunction
\end{algorithmic}
\end{algorithm}
\end{minipage}
\end{wrapfigure}
The algorithm recursively traverses the original game tree in a post-order depth-first fashion: for each traversed node, corresponding nodes are instantiated in the converted game. The chance, terminal, and adversary nodes are copied unaltered since the coordinator player $t$ has only access to the public information visible to the team members. 
Each team member node is instead mapped to a new coordinator node, in which he plays a prescription $\Gamma$ among the combination of possible actions for each information state $I$ belonging to the public team state. % of the node to convert. 
Given a public state $S$, the coordinator issues to the players different recommendations for every possible information set belonging to $S$.
As an example, consider to be in a state of the original game in which one of the team players has to choose between the set of actions $\mathcal{A} = \{a_1,...,a_n\}$ and that her private state can be either $S_1$ or $S_2$. 
Then, the recommendation issued by the coordinator is a tuple that specifies one action for each possible private state (\emph{e.g.,} a valid recommendation would be to play action $a_i$ when the private state is $S_1$ and to play action $a_j$ when the private state is $S_2$).
Once the original ATG has been transformed into a 2p0s game, we solve the latter with state-of-the-art algorithms for this class of games (\emph{e.g.,} CFR, CFR+, MCCFR, DeepCFR \cite{Zinkevich2007RegretMI,Tammelin2014SolvingLI,Lanctot2009MonteCS,Brown2019DeepCR}).
We show an example of a two-player team game and its converted game in Appendix~\ref{app:figures}.
To prove that our game transformation is effective, we need to show the equivalence between a Nash Equilibrium in the converted game and the TMEcor in the original vEFG. 
The proofs can be found in Appendix~\ref{app:proofs}
\begin{restatable}{theorem}{equivalenceNE}
Given a public-turn-taking vEFG $\mathcal G$, and the correspondent converted game $\mathcal G' = \algoname{ConvertGame}(\mathcal G)$, a Nash Equilibrium $\mu_t^*$ in $\mathcal G'$ corresponds to a TMECor $\mu_\team^* = \sigma(\mu_t^*)$ in $\mathcal G$.
\label{th:main}
\end{restatable}
Consider now the specific case in which team members have \textit{common external information}, that means that team members have the same observation of the actions made by adversary and chance players, formally, $\nexists a \in \mathcal A_c \cup \mathcal A_o: Pub_{\mathcal T}(a) = hidden$. The only source of imperfect recallness is the non-observation over a team member's action. To deal with this case, we can resort to \cite{Kaneko1995BehaviorSM}, where the authors provide a polynomial-time algorithm to find an equilibrium. 
As in their case, our conversion algorithm outputs a game tree that has a number of nodes comparable to the one of the original game, thus enabling efficient approximation of the equilibrium. 
This happens because, at the decision nodes of the coordinator, we do not have the combinatorial explosion deriving from the presence of different possible private states as players do not have private information.
%

%We match their results in our framework as well, since, without private information, the coordinator directly suggests one action to be played.\textcolor{blue}{questa ultima frase non si capisce, e' troppo fumosa --- non si capisce il significato di match their results}
%
%\input{content/6_pruning}
\section{Information-lossless Abstractions}\label{sec:pruning}

The conversion procedure presented above allows us to prove the equivalence between the original ATG and its converted version. A crucial question is whether we can produce a reduced, more compact variant of our representation. 
The main reason is that, indeed, applying Algorithm~\ref{alg:PubTeam} to a public state in the original game with $A$ actions for $S$ possible private states results in $A^S$ prescriptions, producing a tree with an exponentially large fan-out.
% While the computation of a TMECor has an unavoidable exponential complexity due to the \textsf{NP}-hardness of the problem\ma{CITA}, any 2p0s solving algorithm can strongly benefit from any improvement made to the game size. 

We propose three algorithms, each returning an information-lossless abstraction of the converted game to attenuate the computational burden of the converted game allowing the application of existing algorithms for equilibrium computation or approximation.

\textbf{Pruned Representation}. 
Whenever a prescription is issued to a team member that has to play a public action, the played action can be used to exclude the private states for which a different action has been given. 
This practice allows us to safely exclude a subset of the possible private states reducing the number of possible prescriptions in successive nodes.
%
% This pruning effect is particularly effective when the chosen prescriptions are various, suggesting many different actions; on the other hand, uninformative prescriptions associating the same action to all private states do not allow the exclusion of any private state.

Such a representation can be obtained in a \textit{online} fashion (\emph{i.e.}, without the need for generating the complete tree and subsequently pruning it) by modifying Algorithm~\ref{alg:PubTeam}. In particular, we add a parameter $\mathcal X$ in \algoname{PubTeamConv} storing the excluded private states for the team. Those private states are the ones for which an action was prescribed but then a different action has been played. By excluding every information set in $\mathcal X$ when building the prescription in Line~\ref{lst:line:prescription}, we can effectively shrink the number of private states to be considered by the coordinator.

An algorithm to recursively convert an ATG into its pruned representation is presented in Appendix~\ref{app:pruned_PubTeam}.

\textbf{Folding Representation}. Consider the case of the extraction of a specific chance outcome when this is private information for a team member and unseen by the adversary (\emph{e.g.}, a private card in poker). We can avoid adding this node to the converted game and instead maintain the probability that each player is in a specific private state, given that some may be excluded as in the pruned representation. Then, we modify the chance node $h''$ after each prescription to sample an action according to the sum of the probabilities of the private states for which that action has been prescribed.
%In the original game, chance outcomes are explicitly represented independently by the visibility of the outcomes, branching the game tree into different subgames according to the specific outcome. Consider the case of the extraction of a specific outcome when this is private information for a team member and unseen by the adversary. In the converted game, such an outcome is not visible to any player, and can therefore be safely postponed as long as no specific action depends on it. The folding representation takes advantage of this property to avoid sampling these types of private states and instead samples an action from the prescription depending on the probability (also called belief) that at a certain point in the game a specific private state is present given the previous actions of all players. The dummy chance nodes $h''$ instantiated in Algorithm \ref{alg:PubTeam} therefore may present different actions each with a probability that is the sum of the probability of the private states for which that action has been prescribed. This representation also takes advantage of the computed beliefs to reduce the prescriptions given by the coordinator as in the pruning representation.
%
Overall, the game size is reduced in case many private states for the players are sampled. Such a representation is similar to the public belief state representation proposed in \cite{Brown_Bakhtin_Lerer_Gong_2020}. The name folding representation derives from the fact that trajectories with the same public actions but different private states are folded one over the other in the converted game.

\textbf{Safe Imperfect Recallness}. 
Whenever a team player is in a state with three or more actions available, a specific action is played and some possible private states are excluded from the belief. 
The specific actions prescribed for the excluded states are not important to describe the information state of the player; we can therefore forget part of the prescription and have imperfect recall among different prescriptions with different prescribed actions for excluded states. 
While this abstraction technique does not directly reduce the number of nodes, it reduces the number of information sets, simplifying the information structure of the game. 
This reduces the space requirements to represent the strategies and makes the algorithms converge faster. 
This abstraction technique is theoretically sound, and the convergence properties of CFR in this imperfect-recall setting have already been addressed by \cite{Lanctot_Gibson_Burch_Zinkevich_Bowling_2012}.

Examples of the application of the techniques described above are shown in Appendix~\ref{app:figures}. 
\section{Experimental Evaluation}\label{sec:experiments}
\begin{table*}[!b]
\begin{subtable}{.45\textwidth}
  \resizebox{\columnwidth}{!}{%
\begin{tabular}{rrrrr}
\toprule
$H$  & normal   & basic    & pruning  & folding  \\
\midrule
%1  & 8.00E+00 & 3.00E+00 & 3.00E+00 & 9.00E+00 \\
2  & 6.40E+01 & 2.70E+01 & 2.70E+01 & 7.50E+01 \\
%3  & 5.12E+02 & 2.19E+02 & 1.35E+02 & 3.75E+02 \\
4  & 4.10E+03 & 1.76E+03 & 5.19E+02 & 1.46E+03 \\
%5  & 3.28E+04 & 1.40E+04 & 1.72E+03 & 4.86E+03 \\
6  & 2.62E+05 & 1.12E+05 & 5.18E+03 & 1.48E+04 \\
%7  & 2.10E+06 & 8.99E+05 & 1.46E+04 & 4.18E+04 \\
8  & 1.68E+07 & 7.19E+06 & 3.92E+04 & 1.13E+05 \\
%9  & 1.34E+08 & 5.75E+07 & 1.01E+05 & 2.93E+05 \\
10 & 1.07E+09 & 4.60E+08 & 2.55E+05 & 7.40E+05 \\
%11 & 8.59E+09 & 3.68E+09 & 6.27E+05 & 1.82E+06 \\
12 & 6.87E+10 & 2.95E+10 & 1.51E+06 & 4.41E+06 \\
%13 & 5.50E+11 & 2.36E+11 & 3.59E+06 & 1.05E+07 \\
14 & 4.40E+12 & 1.88E+12 & 8.40E+06 & 2.46E+07 \\
\bottomrule
\end{tabular}
}
  \caption{only $P1$ has private information}
  \label{fig:sub1}
\end{subtable}
\begin{subtable}{.45\textwidth}
  \resizebox{\columnwidth}{!}{%
\begin{tabular}{rrrrr}
\toprule
   $H$ &   normal &    basic &  pruning &  folding \\
\midrule
   %1 & 8.00E+00 & 9.00E+00 & 9.00E+00 & 9.00E+00 \\
   2 & 6.40E+01 & 8.10E+01 & 8.10E+01 & 7.50E+01 \\
   %3 & 5.12E+02 & 6.57E+02 & 4.05E+02 & 3.75E+02 \\
   4 & 4.10E+03 & 5.26E+03 & 1.56E+03 & 1.46E+03 \\
   %5 & 3.28E+04 & 4.21E+04 & 5.16E+03 & 4.86E+03 \\
   6 & 2.62E+05 & 3.37E+05 & 1.55E+04 & 1.48E+04 \\
   %7 & 2.10E+06 & 2.70E+06 & 4.37E+04 & 4.18E+04 \\
   8 & 1.68E+07 & 2.16E+07 & 1.17E+05 & 1.13E+05 \\
   %9 & 1.34E+08 & 1.73E+08 & 3.04E+05 & 2.93E+05 \\
  10 & 1.07E+09 & 1.38E+09 & 7.65E+05 & 7.40E+05 \\
  %11 & 8.59E+09 & 1.10E+10 & 1.88E+06 & 1.82E+06 \\
  12 & 6.87E+10 & 8.84E+10 & 4.53E+06 & 4.41E+06 \\
  %13 & 5.50E+11 & 7.07E+11 & 1.08E+07 & 1.05E+07 \\
  14 & 4.40E+12 & 5.65E+12 & 2.52E+07 & 2.46E+07 \\
\bottomrule
\end{tabular}
}
  \caption{both $P1$ and $P2$ have private information}
  \label{fig:sub2}
\end{subtable}
\caption{Comparison of total number of nodes for $C=3$, $A=2$.}
\label{fig:pruning}
\end{table*}
\begin{figure*}[!b]
\centering
\begin{subfigure}{0.32\textwidth}
  \centering
  \includegraphics[width=\textwidth]{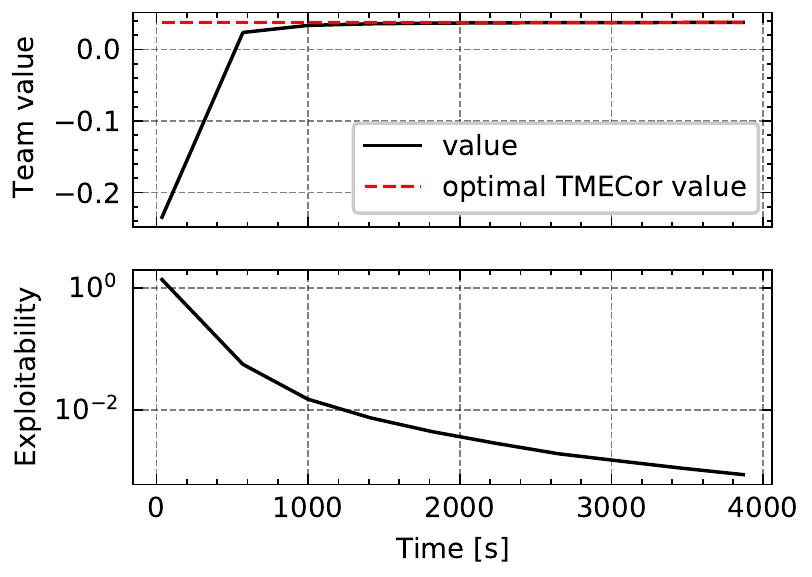}
  
  \small{Peak RAM occupation: 96 MB}
  \caption{3-player Kuhn poker, with 4 card ranks and adversary playing first.}
  \label{fig:poker1}
\end{subfigure}%
\hspace*{0.01\textwidth}%
\begin{subfigure}{.32\textwidth}
  \centering
  \includegraphics[width=\textwidth]{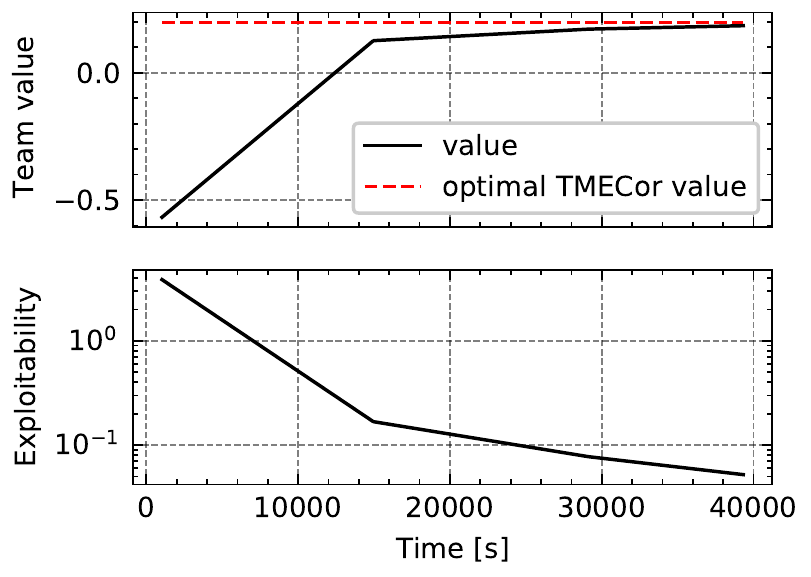}
  
  \small{Peak RAM occupation: 912 MB}
  \caption{3-player Leduc poker, with 3 card ranks, 1 raise, and adversary playing first.}
  \label{fig:poker2}
\end{subfigure}
\hspace*{0.01\textwidth}%
\begin{subfigure}{.32\textwidth}
  \centering
  \includegraphics[width=\textwidth]{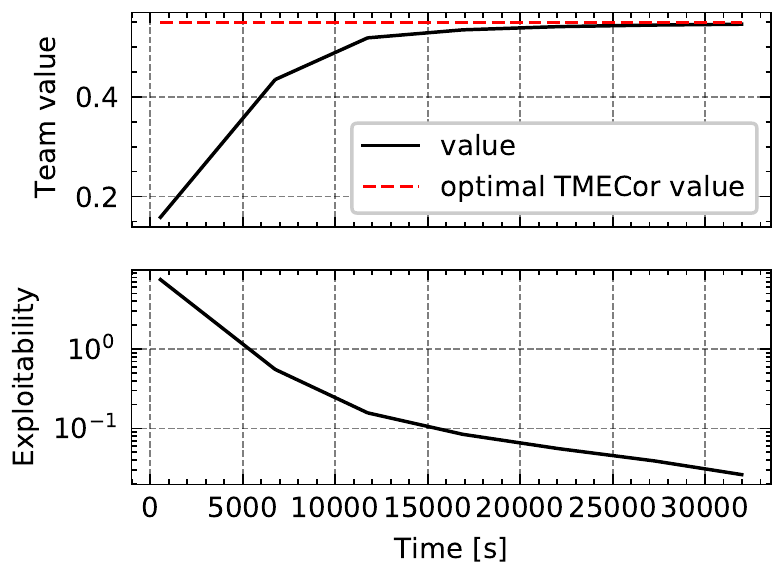}
  
  \small{Peak RAM occupation: 599 MB}
  \caption{3-player Leduc poker, with 2 card ranks, 2 raises, and adversary playing first.}
  \label{fig:poker3}
\end{subfigure}
\caption{Performances of Linear CFR+ applied to Poker game instances.}
\end{figure*}
In this section, we evaluate the benefits of the lossless abstraction techniques when applied to a toy game, and present the application of our conversion procedure to multiplayer instances of Kuhn and Leduc poker.
Details on the experimental setting and further results are shown in Appendix~\ref{app:exp}.
\subsection{Impact of Abstraction Techniques}
To evaluate the impact of information-lossless abstraction techniques, we designed a simple parametric game for which we are able to determine in closed-form the total number of nodes.
For the sake of simplicity, we can ignore the presence of an opponent (recall that opponent and chance nodes are copied as they are, hence they do not have an impact on the increase in size of the converted game).
The game we considered is a two-player game $\mathcal G$, with players $P1$ and $P2$, in which there are $C$ chance outcomes at root, observed by $P1$ and not by $P2$. 
After observing the outcome of the chance node, $P1$ acts sequentially $H$ times, each time choosing between $A$ actions.
$P2$ does not observe the actions performed by $P1$ and acts at the last round of the game, choosing between $A$ actions. 

We evaluate the size of the games obtained by applying the basic algorithm and the sizes of the games obtained through the abstraction techniques presented in Section~\ref{sec:pruning} and denoted as \emph{pruning} and \emph{folding}. 
Furthermore, as a baseline, we show the number of normal-form plans. 
Table~\ref{fig:pruning} shows the results for $C=3$, $A = 2$ and for various values of the parameter $H$\footnote{The complete results are shown in Appendix~\ref{app:pruning}.}.
The derivation of the equation that shows the total number of nodes is shown in Appendix~\ref{app:pruning}.

% comment
We observe that the \emph{pruning} technique is particularly effective at dampening the exponential factor due to the combinatorial structure of prescriptions.%, thanks to the belief-based pruning.
Moreover, the \emph{folding} technique combines the benefits of the \emph{pruning} technique, with the specific tradeoff imposed by the use of delayed chances. 
Indeed, all the histories belonging to the same public state are represented in a single node, at the cost of a chance node added after each prescription. 
This tradeoff is useful for game states in which many private states are possible. 
On the other hand, when the possible private states are reduced to one or two, the extra chance nodes increase the size of the game more than in the \emph{pruning} representation.
%
% We decide to use the folding representation as our focus in poker instances with an information structure similar to Table \ref{fig:sub2}. \fe{In che senso?}

\subsection{Application to Team Poker Games}
To test the effectiveness of our approach in a real case, we apply our conversion procedure to the multiplayer versions of Kuhn poker and Leduc poker. 
The Kuhn instances we use are parametric in the number of ranks available, and on whether the adversary plays first, second or third in the game. 
Similarly, Leduc instances are parametric on the number of ranks, on the position of the adversary, but also on the number of raises that can be made.
Figures~\ref{fig:poker1}, ~\ref{fig:poker2}, ~\ref{fig:poker3} show the convergence performances on the converted games of Linear CFR+ \cite{Brown2019SolvingIG}. %
We adopted three-player versions of the two games in which there are two players \emph{colluding} against the other.
The converted games were implemented using the pure python API provided by Open Spiel~\cite{lanctot2019openspiel}. We took advantage of the already available 2p0s solvers.
More details on the experimental setting are available in Appendix \ref{app:exp}, while Appendix \ref{app:sizes} shows the sizes of the converted games for a varying number of ranks and raises.

We report a paired plot showing both the value and exploitability convergence, along with the optimal value of a TMECor as computed by \cite{Farina2021ConnectingOE}, represented as horizontal dashed lines in the team value plot.
The experimental performances are coherent with the theoretical result of Theorem~\ref{th:main}; the Nash Equilibrium returned by the algorithm has the same value of the TMECor of the original game in all the tests while the exploitability values of the average strategy converge to $0$, indicating convergence to the equilibrium.
We remark that the time spent by Linear CFR+ applied to our converted game is higher than the running time required by the algorithm presented  in~\cite{Farina2021ConnectingOE}. This is due to the use of different programming languages for the two algorithms (as aforementioned, our algorithms are coded in Python, whereas the algorithms in~\cite{Farina2021ConnectingOE} are coded in C++ and use commercial solvers), while the question of whether even more efficient abstractions can be formulated is still open.
%a potential inefficiency of our abstractions, which could be made much more compact without losing information.
%
\section{Conclusions}\label{sec:conclusion}
We presented a conversion procedure to transform an adversarial team game in a 2p0s game and proved that a NE in the converted game corresponds to a TMECor in the original game. This conversion procedure is based on the public information among team members, which allows to remove each player's private information in the converted representation at the cost of a exponential increase in size due to the \textsf{NP}-hardness of computing a TMECor in a generic team game.
We propose three pruning techniques to obtain a resulting tree which is more computationally tractable.

This conversion retains the public structure of the original game, and allows the application of more scalable techniques for the resolution of adversarial team games, such as abstractions and continual resolving. 
Future works leveraging our techniques in the converted game have therefore the opportunity to overcome current state of the art techniques on larger game instances.
Another direction is the development of a conversion procedure to produce a pruned or folded representation in online fashion, avoiding the construction of the entire tree.
Additionally, in principle, our method can be adopted also in $n$ players vs $m$ players games. In such scenarios the support at the equilibrium is no longer small and mathematical programming algorithms loose some efficiency. 
Solving team vs team games with a public information approach could therefore be a very interesting future direction. 
\clearpage
\bibliographystyle{abbrvnat}
\bibliography{neurips_2021}
%
%%%%%%%%%%%%%%%%%%%%%%%%%%%%%%%%%%%%%%%%%%%%%%%%%%%%%%%%%%%%
%\input{content/checklist}
%%%%%%%%%%%%%%%%%%%%%%%%%%%%%%%%%%%%%%%%%%%%%%%%%%%%%%%%%%%%

\clearpage
\appendix
\section{Proofs} \label{app:proofs}
\sizepublicturn*
\begin{proof}[Sketch of the Proof]
A very simple procedure that allows us to prove the lemma is the following: we can set for each level of the converted game a player, by cycling through all players (chance included). Then, we can add all the histories of the original game one by one, while imposing that at each level only the chosen player can play. If the history has no action assigned to the level's player, then we can add a dummy player node, with only a noop operation, and try to prosecute with the actions of the original history in the next node. The visibility of the noop actions is unseen for all players apart from the one playing.

This procedure guarantees to get a strategically equivalent game by adding at most $\mathcal O((|\mathcal N| + 1)|\mathcal H|)$ for any of the $|\mathcal H|$ histories in the original game. This proves that the number of histories in the converted game is $\mathcal O((|\mathcal N| + 1)|\mathcal H|^2)$
\end{proof}

The following lemmas are instrumental to derive the main result. In the following, we make use of reduced normal form plans, and we refer to them as plans and pure strategies for clarity. We also drop the superscript ``$^\star$'' to simplify the notation.
\begin{restatable}{lemma}{lemmaunoeq}
Given a public-turn-taking vEFG $\mathcal G$, and the correspondent converted game $\mathcal G' = \algoname{ConvertGame}(\mathcal G)$, each joint pure strategy $\pi_{\mathcal T}$ in $\mathcal G$ can be mapped to a strategy $\pi_t$ in $\mathcal G'$, such that the traversed histories have been mapped by \algoname{PublicTeamConversion}. Formally, $\forall \pi_{\mathcal T} \; \exists \pi_t$ such that $\forall \pi_o,\pi_c$:

\begin{equation*}
\begin{array}{c}
(\algoname{PubTeamConversion}(h))_{h \text{ reached by playing } (\pi_{\mathcal T}, \pi_o, \pi_c) \text{ in } \mathcal G }\\
\equiv\\
(h')_{h'\text{ reached by playing } (\pi_t, \pi_o, \pi_c) \text{ in } \mathcal G' }. 
\end{array}
\end{equation*}
\label{th:lemma1}
\end{restatable}
\begin{proof}
We can prove Lemma \ref{th:lemma1} recursively by traversing both $\mathcal G$ and $\mathcal G'$ while constructing the equivalent pure strategy in the converted game. We start by $h_\emptyset$ and $h_\emptyset'$. We know that $h_\emptyset' = \algoname{PublicTeamConversion}(h_\emptyset)$.

Let $h$ and $h'$ be the nodes currently reached recursively in $\mathcal G$ and $\mathcal G'$, such that $h' = \algoname{PublicTeamConversion}(h)$, with the guarantee that correspondent histories in the trajectories traversed up to this point in the two games have such a property. We thus have the guarantee that $h$ and $h'$ are both terminal or both share the same player.
Then:
\begin{itemize}
\item Case \textbf{team member node}

Let $a = \pi_{\mathcal T}[\mathcal I (h)]$ be the action specified by $\pi_{\mathcal T}$ to be taken at $\mathcal I(h)$. We can construct a prescription $\Gamma = (\pi_{\mathcal T}[I])_{I \in \mathcal S[h]}$ equivalent to the pure strategy $\pi_{\mathcal T}$ in this public state. We set $\pi_t[\mathcal I'(h')] = \Gamma$, and prosecute our proof from the two reached nodes $h'\Gamma a$ and $ha$. The construction procedure \algoname{PublicTeamConversion} guarantees In fact that $h'\Gamma a = \algoname{PubTeamConversion}(ha)$.

\item Case \textbf{chance or opponent node}

$\pi_o$ and $\pi_c$ are common to both the traversals. This guarantees that the action $a$ suggested by the policy is equal, and by construction of the conversion procedure $h'a' = \algoname{PubTeamConversion}(ha)$. We can thus proceed with the proof considering $h'a$ and $ha$.

\item Case \textbf{terminal node}

By construction, they have the same value for all players. This concludes the recursive proof.
\end{itemize}
\end{proof}
\begin{restatable}{lemma}{lemmadueeq}
Given a public-turn-taking vEFG $\mathcal G$, and the corresponding converted game $\mathcal G' = \algoname{ConvertGame}(\mathcal G)$, each coordinator pure strategy $\pi_t$ in $\mathcal G'$ can be mapped to a strategy $\pi_{\mathcal T}$ in $\mathcal G$, such that the traversed histories have been mapped by \algoname{PublicTeamConversion}. Formally, $\forall \pi_t \; \exists \pi_{\mathcal T}$ such that $\forall \pi_o,\pi_c$:
\begin{align*}
\begin{array}{c}
(\algoname{PubTeamConversion}(h))_{h \text{ reached by playing } (\pi_{\mathcal T}, \pi_o, \pi_c) \text{ in } \mathcal G }\\
\equiv\\
(h')_{h' \text{ reached by playing } (\pi_t, \pi_o, \pi_c) \text{ in } \mathcal G'} .
\end{array}
\end{align*}
\label{th:lemma2}
\end{restatable}
\begin{proof}
We can prove Lemma \ref{th:lemma2} recursively by traversing both $\mathcal G'$ and $\mathcal G$ while constructing the equivalent pure strategy in the original game. We start by $h_\emptyset'$ and $h_\emptyset$. We know that $h_\emptyset' = \algoname{PublicTeamConversion}(h_\emptyset)$.

Let $h'$ and $h$ be the nodes currently reached recursively in $\mathcal G'$ and $\mathcal G$, such that $h' = \algoname{PublicTeamConversion}(h)$, and with the guarantee that correspondent histories in the trajectories traversed in the two games have such a property. We thus have the guarantee that $h$ and $h'$ are both terminal or both share the same player. Then:
\begin{itemize}
\item Case \textbf{team member node}

Let $\Gamma = \pi_t[\mathcal I(h')]]$ be the prescription specified by $\pi_t$ to be taken at $\mathcal I'(h')$. We can extract the prescribed action $a = \Gamma[I]$ to be played in history $h$. We set $\pi_{\mathcal T}[\mathcal I(h)] = a$, and prosecute our proof from the two reached nodes $h'\Gamma a$ and $ha$. The \algoname{PublicTeamConversion} procedure guarantees In fact that $h'\Gamma a = \algoname{PubTeamConversion}(ha)$.

\item Case \textbf{chance or opponent node}

$\pi_o$ and $\pi_c$ are common to both the traversals. This guarantees that the action $a$ suggested by the policy is equal, and by construction of the conversion procedure $h'a' = \algoname{PubTeamConversion}(ha)$. We can thus proceed with the proof considering $h'a$ and $ha$.

\item Case \textbf{terminal node}

By construction, they have the same value for all players. This concludes the recursive proof.
\end{itemize}
\end{proof}
Thanks to Lemmas~\ref{th:lemma1}~and~\ref{th:lemma2}, we can define the following functions that are used to map strategies from the original game to the converted game: 
\begin{definition}[Mapping functions among original and converted game]
We define:
\begin{itemize}
\item $\rho: \Pi_{\mathcal T} \to \Pi_t$ is the function mapping each $\pi_{\mathcal T}$ to the $\pi_t$ specified by the procedure described in the proof of Lemma~\ref{th:lemma1}.
\item $\sigma: \Pi_t \to \Pi_{\mathcal T}$ is the function mapping each $\pi_t$ to the $ \pi_{\mathcal T}$ specified by the procedure described in the proof of Lemma~\ref{th:lemma2}.
\end{itemize}

Those two functions can also be extended to mixed strategies, by converting each pure plan and summing the probability masses of the converted plans. Formally, we have:
\begin{align*}
\forall \mu_{\mathcal T} \in \Delta^{\Pi_{\mathcal T}} : \rho(\mu_{\mathcal T})[\pi_t] = \sum_{\pi_{\mathcal T} : \rho(\pi_{\mathcal T})=\pi_t} \mu_{\mathcal T}(\pi_{\mathcal T}),\\
\forall \mu_t \in \Delta^{\Pi_t} : \sigma(\mu_t)[\pi_{\mathcal T}] = \sum_{\pi_t : \sigma(\pi_t)=\pi_{\mathcal T}} \mu_t(\pi_t).
\end{align*}
\end{definition}
We can now show the payoff-equivalence between a game $\mathcal{G}$ and its converted version $\mathcal G' = \algoname{ConvertGame}(\mathcal G)$: 
\begin{restatable}{theorem}{theoremequivalence}
A public-turn-taking vEFG $\mathcal G$ and its converted version $\mathcal G' = \algoname{ConvertGame}(\mathcal G)$ are payoff-equivalent. Formally:
\begin{align*}
\forall \pi_{\mathcal T} \; \forall \pi_o,\pi_c :  u_{\mathcal T}(\pi_{\mathcal T}, \pi_o, \pi_c) = u_t(\rho(\pi_{\mathcal T}), \pi_o, \pi_c), \\
\forall \pi_t \; \forall \pi_o,\pi_c :  u_{\mathcal T}(\sigma(\pi_t), \pi_o, \pi_c) = u_t(\pi_t, \pi_o, \pi_c).
\end{align*}
\end{restatable}
\begin{proof}
The proof follows trivially from Lemmas~\ref{th:lemma1}~and~\ref{th:lemma2}.
\end{proof}

Furthermore, the correspondence between the strategies in the two games are then used to derive the main result of this work that shows the equivalence between a NE of the converted 2p0s game and a TMEcor of the original ATG, thus enabling the usage of a powerful set of tools to compute equilibria for ATGs.  
\begin{restatable}{theorem}{equivalenceNE}
Given a public-turn-taking vEFG $\mathcal G$, and the correspondent converted game $\mathcal G' = \algoname{ConvertGame}(\mathcal G)$, a Nash Equilibrium $\mu_t^*$ in $\mathcal G'$ corresponds to a TMECor $\mu_\team^* = \sigma(\mu_t^*)$ in $\mathcal G$.
\label{th:main}
\end{restatable}
\begin{proof}
By hypothesis that $\mu_t^*$ is a NE, we have that:
\begin{equation*}
\mu_t^* \in \argmax_{\mu_t \in \Delta^{\Pi_t}} \min_{\mu_o \in \Delta^{\Pi_o}} \sum_{\substack{\pi_t \in \Pi_t \\ \pi_o\in \Pi_o \\ \pi_c \in \Pi_c}} \mu_t(\pi_t) \mu_o(\pi_o) \mu_c(\pi_c) u_t(\pi_t, \pi_o, \pi_c)
\end{equation*}
We need to prove:
\begin{equation*}
\sigma(\mu_t^*) \in \argmax_{\mu_{\mathcal T} \in \Delta^{\Pi_{\mathcal T}}} \min_{\mu_o \in \Delta^{\Pi_o}} \sum_{\substack{\pi_{\mathcal T} \in \Pi_{\mathcal T} \\ \pi_o\in \Pi_o \\ \pi_c \in \Pi_c}} \mu_{\mathcal T}(\pi_{\mathcal T}) \mu_o(\pi_o) \mu_c(\pi_c) u_{\mathcal T}(\pi_{\mathcal T}, \pi_o, \pi_c)
\end{equation*}

Let $\min_{TMECor}(\mu_{\mathcal T})$ and $\min_{NE}(\mu_t)$ be the inner minimization problem in the TMECor and NE definition respectively.

\textit{Absurd.} Suppose $\exists \; \bar \mu_{\mathcal T}$ with a greater value than $\sigma(\mu_t^*)$. Formally:
\begin{equation*}
\min_{TMECor}(\bar \mu_{\mathcal T}) > \min_{TMECor}(\mu_t^*).
\end{equation*}
In such a case, we could define $\bar \mu_t = \rho(\bar \mu_{\mathcal T})$ having value:
\begin{equation*}
\min_{NE}(\bar \mu_t) = \min_{TMECor}(\bar \mu_{\mathcal T}) > \min_{TMECor}(\sigma(\mu_t^*)) = \min_{NE}(\mu_t^*),
\end{equation*}
where the equalities are due to the payoff equivalence. However this is absurd since by hypothesis $\mu_t^*$ is a maximum. Therefore necessarily:
\begin{equation*}
\sigma(\mu_t^*) \in \argmax_{\mu_{\mathcal T} \in \Delta^{\Pi_{\mathcal T}}} \min_{NE}(\mu_{\mathcal T}).
\end{equation*}
This concludes the proof.
\end{proof}
\section{Information Structure in Team Games} \label{app:info_structure}
The core problem of finding a TMECor in adversarial team games resides in \textit{asymmetric visibility} since team members have a private state that does not allow to create a joint coordination player.

In the following, we characterize the possible types of asymmetric visibility that may cause imperfect recall for the joint player, and singularly address them.

\begin{itemize}
\item \textbf{Non-visibility over a team member's action}. If a team member plays an action that is hidden from another team member, the joint team player would have imperfect recall due to forgetting his own played actions. This source of imperfect recallness can be avoided in a TMECor by considering the shared deterministic strategies before the game starts. This allows us to know a priori which are the exact actions played by team members in each node. Thus it is safe to apply a perfect recall refinement in the original game, which corresponds to always consider the chosen action of a team member as $\mathrm{seen}$ by other team members.

\item \textbf{Non-visible game structure}. Consider two nodes in the same information set for a player before which the other team member may have played a variable number of times, due to a chance outcome non-visible to the team member of these nodes. In this case, a perfect recall refinement is not applicable to distinguish the nodes, because it would give the joint coordinator visibility not correspondent to the one of the players in the game. To solve this edge case, we require the property of public turn-taking.

\item \textbf{Private information disclosed by chance/adversary to specific team members.} It is the most complex type of non-visibility, since in a TMECor we have no explicit communication channels through which to share information, and therefore this type of joint imperfect recall can only be addressed by considering a strategically equivalent representation of the game in which at most one of the team players has private information. 
%In fact, this specific type of non-visibility is the one avoided by SIMS since no perfect recall refinement can perform such a transformation and it is solved through the normal form representation of a team player in HCG, FTP, FCG since the choice of a plan is independent of the private information discovered during the game.
\end{itemize}

\section{Pruned Public Team Conversion}\label{app:pruned_PubTeam}
In the following, we present a variation of Algorithm~\ref{alg:PubTeam} that directly produces an information-lossless abstraction of the converted game. Modifications to Algorithm~\ref{alg:PubTeam} are highlighted in bold.

\begin{algorithm}[!htb]
\caption{Pruned Public Team Conversion}\label{alg:pruned_PubTeam}
\begin{algorithmic}[1]
	\Function {ConvertGame}{$\mathcal G$}
		\State initialize $\mathcal G'$ new game
		\State $\mathcal N' \gets \{t,o\}$
		\State $h'_{\emptyset} \gets$ \Call{PubTeamConv}{$h_{\emptyset}, \mathcal G, \mathcal G'$} \Comment{new game root}
		\State \Return{$\mathcal G'$}
	\EndFunction
	\item[] % to skip line
	\Function{PubTeamConv}{$h$, $\mathcal G$, $\mathcal G'$, $\boldsymbol{\mathcal X}$}
		\State initialize $h' \in \mathcal H'$
  		\If {$h \in \mathcal Z$} \Comment{terminal node}
  			\State $h' \gets h' \in \mathcal Z'$
  			\State $u'_p(h') \gets u_p(h) \quad \forall p \in \mathcal N$
  		\ElsIf {$\mathcal P(h) \in \{o, c\}$} \Comment{opponent or chance}
  			\State $\mathcal P'(h') \gets \mathcal P(h)$
  			\State $\mathcal A'(h') \gets \mathcal A(h)$
  			\If {$h$ is chance node}
  				\State $\sigma_c'(h') = \sigma_c(h)$
  			\EndIf
  			\For {$a' \in \mathcal A'(h')$}
  				\State $Pub_t'(a') \gets$ check $Pub_{\mathcal T}(a') = \mathrm{pub}$
  				\State $Pub_o'(a') \gets Pub_o(a')$
  				\State $h'a' \gets$ \Call{PubTeamConv}{$ha'$, $\mathcal G$, $\mathcal G'$, $\boldsymbol{\mathcal X}$}
  			\EndFor 
  		\Else \Comment{team member}
  			\State $\mathcal P'(h') = t$
  			\State $\mathcal A'(h') \gets \bigtimes_{I \in \mathcal  S_{\mathcal T}(h) \boldsymbol{: \nexists I' \in \mathcal X \; matching \; I}} \mathcal A(I)$ \label{lst:line:prescription} \Comment{prescriptions}
  			\For {$\Gamma' \in \mathcal A'(h')$}
  				\State $Pub_t'(\Gamma') \gets \mathrm{seen}, Pub_o'(\Gamma') \gets \mathrm{unseen}$
  				\State $a' \gets \Gamma'[\mathcal I(h)]$ \Comment{extract chosen action}
  				\State \boldsymbol{$\mathcal X \gets \mathcal X \cup \{I : \Gamma'(I) \neq a'\}$} \Comment{update $\mathcal X$ with incompatible private states}
  				\State initialize $h'' \in \mathcal H'$
  				\State $\mathcal A'(h'') \gets \{a'\}$
 				\State $\mathcal P(h'') = c$
  				\State $Pub_t'(a') \gets \mathrm{seen}$
 				\State $Pub_o'(a') = Pub_o(a')$
 				\State $\sigma_c'(h'') =$ play $a'$ with probability 1
  				\State $h''a' \gets$ \Call{PubTeamConv}{$ha'$, $\mathcal G$, $\mathcal G'$, $\boldsymbol{\mathcal X}$}
  				\State $h'\Gamma \gets h''$ 
  			\EndFor 
  		\EndIf
  		\State \Return $h'$
	\EndFunction
\end{algorithmic}
\end{algorithm}

\section{Converted game representation} \label{app:figures}
We present a complete game as an example to show the effects of different types of abstractions applied to the converted game. To ease the visualization, we focus on a cooperative game with no adversary.

\begin{figure}[!htb]
    \centering
    \makebox[\textwidth][c]{\includegraphics[scale=0.75]{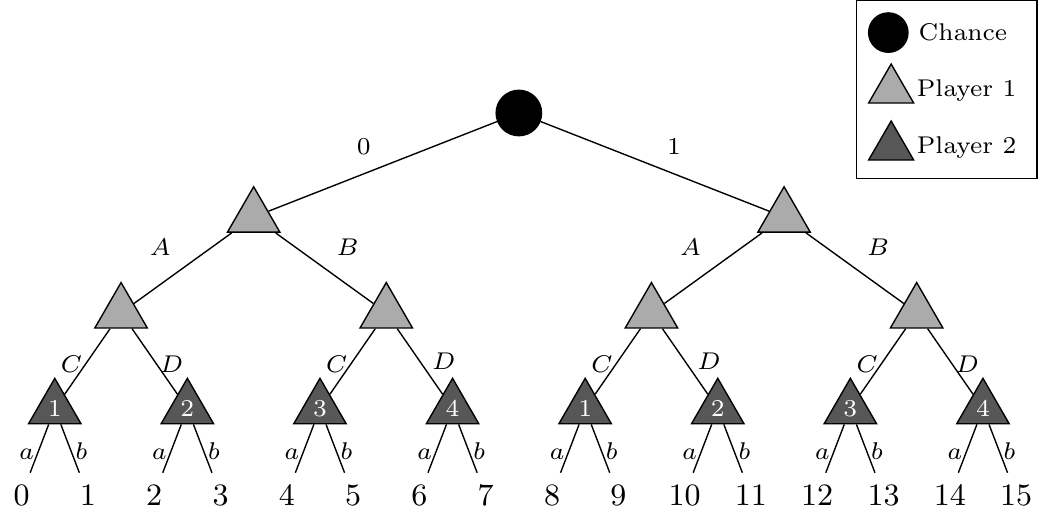}}
    \caption[Example of a cooperative game]{Example of a cooperative game. Player 2 can see all actions apart from chance outcomes $0$, $1$. Nodes of a player with same number are in the same infoset.}
    \label{fig:example}
\end{figure}

\begin{figure}[!htb]
    \centering
    \makebox[\textwidth][c]{\includegraphics[scale=0.75]{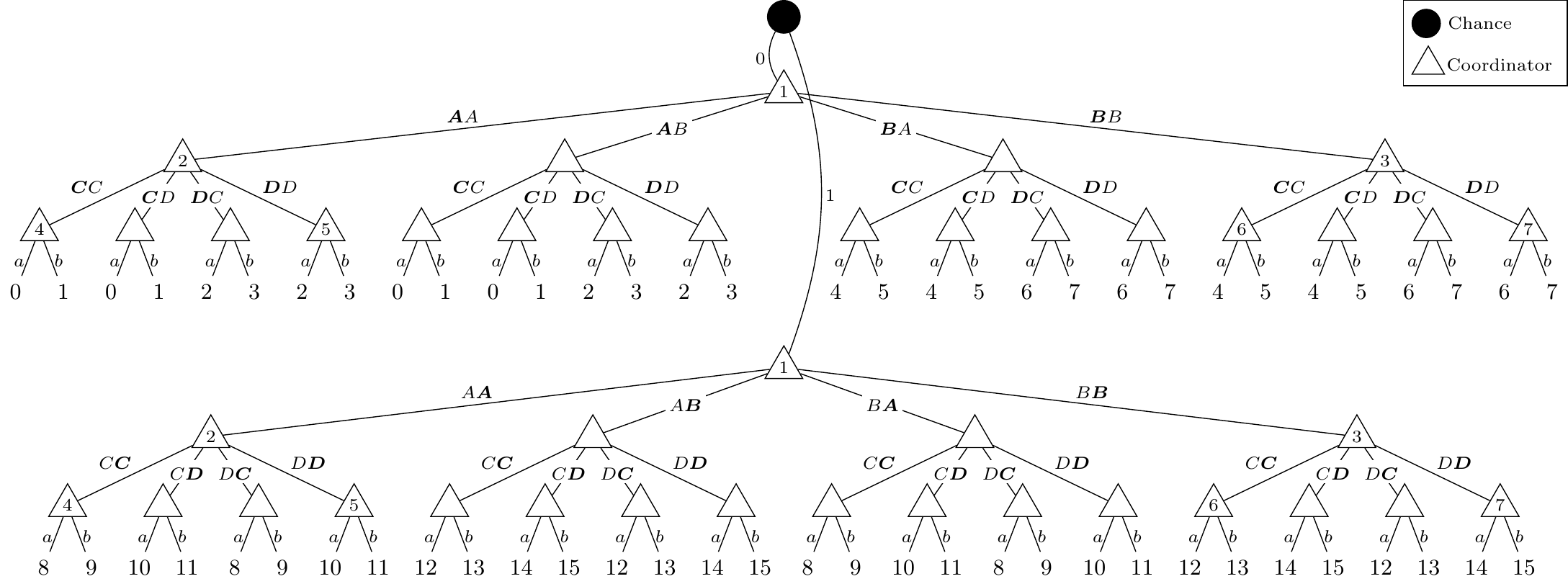}}
    \caption[Example of a converted game]{Example of Figure \ref{fig:example} converted. Nodes of a player with same number are in the same infoset. For notational clarity, dummy chance nodes are not represented, prescriptions list the action to take for private state $0$ and $1$, the action taken afterward is in bold in the prescription.}
    \label{fig:convExample}
\end{figure}

\begin{figure}[H]
    \centering
    \makebox[\textwidth][c]{\includegraphics[scale=0.75]{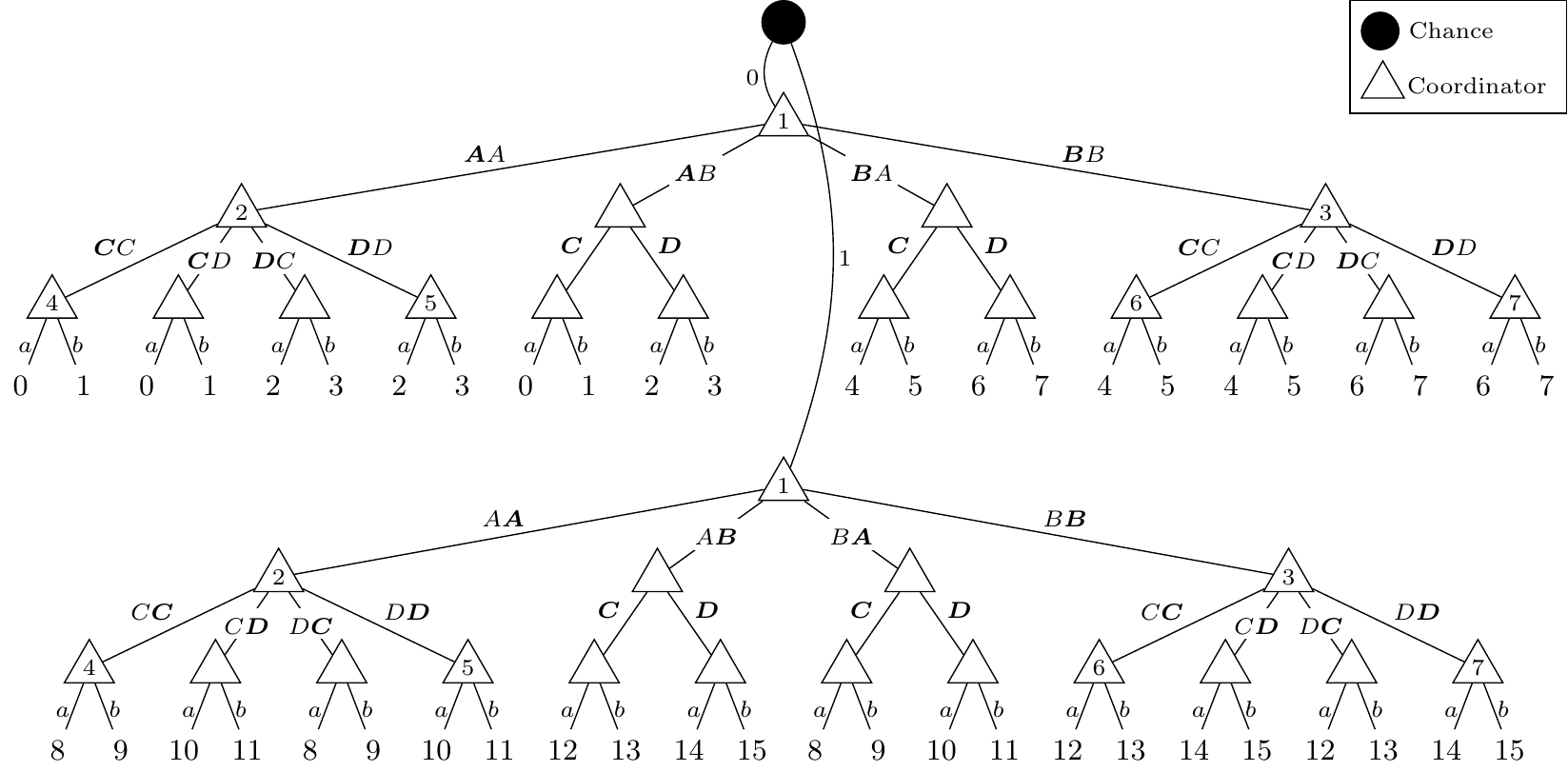}}
    \caption[Example of a converted pruned game]{Example of Figure \ref{fig:example} converted using belief pruning. Nodes of a player with the same number are in the same infoset. For notational clarity, dummy chance nodes are not represented, prescriptions list the action to take for private state $0$ and $1$, the action taken afterward is in bold in the prescription.}
    \label{fig:convPrunedExample}
\end{figure}

\begin{figure}[H]
    \centering
    \makebox[\textwidth][c]{\includegraphics[scale=0.75]{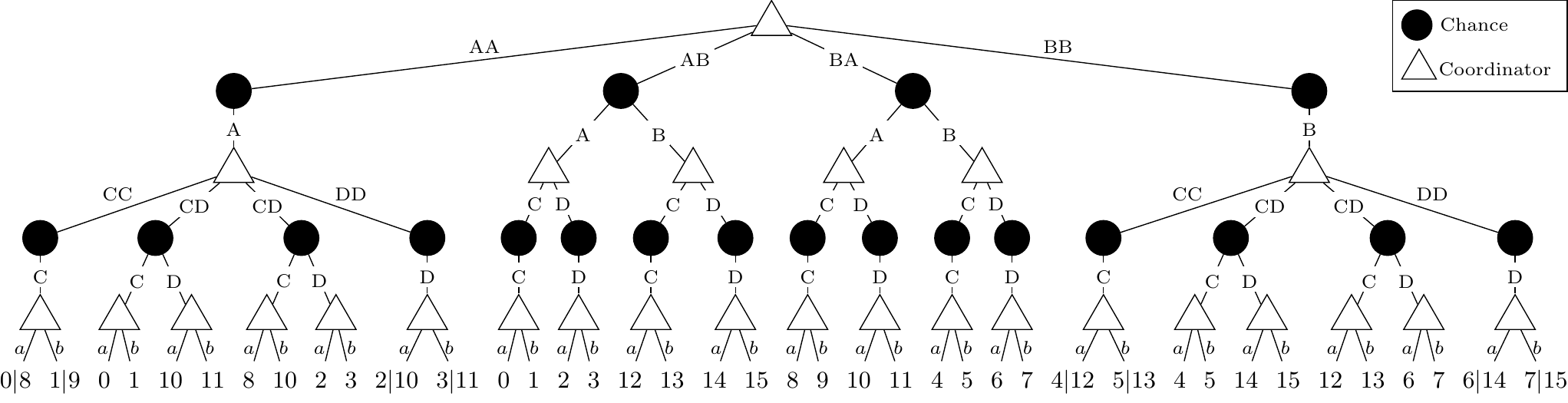}}
    \caption[Example of a converted folded game]{Example of Figure \ref{fig:example} converted using folded representation. For notational clarity, prescriptions list the action to take for private state $0$ and $1$. Terminal nodes in the form $x|y$ represent a terminal node which has a weighted average value with respect to the outcomes $x$ and $y$.}
    \label{fig:convFoldedExample}
\end{figure}

\section{Information-lossless abstraction evaluation} \label{app:pruning}
\begin{table*}[!ht]
\begin{subtable}{.45\textwidth}
  \resizebox{\columnwidth}{!}{%
\begin{tabular}{rrrrr}
\toprule
$H$  & normal   & basic    & pruning  & folding  \\
\midrule
1  & 8.00E+00 & 3.00E+00 & 3.00E+00 & 9.00E+00 \\
2  & 6.40E+01 & 2.70E+01 & 2.70E+01 & 7.50E+01 \\
3  & 5.12E+02 & 2.19E+02 & 1.35E+02 & 3.75E+02 \\
4  & 4.10E+03 & 1.76E+03 & 5.19E+02 & 1.46E+03 \\
5  & 3.28E+04 & 1.40E+04 & 1.72E+03 & 4.86E+03 \\
6  & 2.62E+05 & 1.12E+05 & 5.18E+03 & 1.48E+04 \\
7  & 2.10E+06 & 8.99E+05 & 1.46E+04 & 4.18E+04 \\
8  & 1.68E+07 & 7.19E+06 & 3.92E+04 & 1.13E+05 \\
9  & 1.34E+08 & 5.75E+07 & 1.01E+05 & 2.93E+05 \\
10 & 1.07E+09 & 4.60E+08 & 2.55E+05 & 7.40E+05 \\
11 & 8.59E+09 & 3.68E+09 & 6.27E+05 & 1.82E+06 \\
12 & 6.87E+10 & 2.95E+10 & 1.51E+06 & 4.41E+06 \\
13 & 5.50E+11 & 2.36E+11 & 3.59E+06 & 1.05E+07 \\
14 & 4.40E+12 & 1.88E+12 & 8.40E+06 & 2.46E+07 \\
\bottomrule
\end{tabular}
}
  \caption{only $P1$ has private information}
  \label{fig:sub1*}
\end{subtable}
\begin{subtable}{.45\textwidth}
  \resizebox{\columnwidth}{!}{%
\begin{tabular}{rrrrr}
\toprule
  $H$ &   normal &    basic &  pruning &  folding \\
\midrule
  1 & 8.00E+00 & 9.00E+00 & 9.00E+00 & 9.00E+00 \\
  2 & 6.40E+01 & 8.10E+01 & 8.10E+01 & 7.50E+01 \\
  3 & 5.12E+02 & 6.57E+02 & 4.05E+02 & 3.75E+02 \\
  4 & 4.10E+03 & 5.26E+03 & 1.56E+03 & 1.46E+03 \\
  5 & 3.28E+04 & 4.21E+04 & 5.16E+03 & 4.86E+03 \\
  6 & 2.62E+05 & 3.37E+05 & 1.55E+04 & 1.48E+04 \\
  7 & 2.10E+06 & 2.70E+06 & 4.37E+04 & 4.18E+04 \\
  8 & 1.68E+07 & 2.16E+07 & 1.17E+05 & 1.13E+05 \\
  9 & 1.34E+08 & 1.73E+08 & 3.04E+05 & 2.93E+05 \\
  10 & 1.07E+09 & 1.38E+09 & 7.65E+05 & 7.40E+05 \\
  11 & 8.59E+09 & 1.10E+10 & 1.88E+06 & 1.82E+06 \\
  12 & 6.87E+10 & 8.84E+10 & 4.53E+06 & 4.41E+06 \\
  13 & 5.50E+11 & 7.07E+11 & 1.08E+07 & 1.05E+07 \\
  14 & 4.40E+12 & 5.65E+12 & 2.52E+07 & 2.46E+07 \\
\bottomrule
\end{tabular}
}
  \caption{both $P1$ and $P2$ have private information}
  \label{fig:sub2*}
\end{subtable}
\caption{Comparison of total number of nodes for $C=3$, $A=2$.}
\label{fig:pruning_compl}
\end{table*}
We proceed to analyze the size in total number of coordinator nodes of the converted game for player P1 nodes, depending on the abstraction technique used. To formalize the total number of nodes, we use a succession notation, where $s_l(c)$ indicates the number of nodes at level $l$ in which P1 may be in exactly $c$ private states. Such a notation is particularly useful to represent the relation between private states, abstraction, and the total number of nodes. Table~\ref{fig:pruning_compl} summarizes the results for $C=3$, $A = 2$ and for various values of the parameter $H$.

\medskip
\textbf{Normal form representation}. As a baseline comparison, we compute the number of normal form plans in the game. The total number of plans for P1 can be computed as $A^{C \cdot H}$.

\medskip
\textbf{Basic Representation}. Each of the $H$ levels has $A^C$ actions, and we have $C$ independent trees due to the initial chance. Since we do not perform any belief pruning, all nodes have $C$ possible private states.

The correspondent succession is:
\begin{align*}
x_0(c) & = 
\left\{\begin{array}{ll}
0 & \text{if } c \not = C\\
C & \text{if } c = C
\end{array}\right.
\text{initial C chance outcomes}\\
x_l(c) & =
\left\{\begin{array}{ll}
0 & \text{if } c \not = C\\
x_{l-1}(c) \cdot A^C & \text{if } c = C
\end{array}\right.
A^C \text{ fanout at each level}
\end{align*}
Therefore:
\begin{align*}
\mathrm{tot}(C,A,H) & = \sum_{l=0}^H \sum_{c=1}^C x_l(c) \\
& = C \cdot \sum_{i=1}^H(A^C)^{i}
\end{align*}

\medskip
\textbf{Pruning Representation}. Given a node with a generic number $I$ of private states, we can work by induction to retrieve the number of generated nodes:
\begin{itemize}
\item children left with $I$ possible infostates: $A$. They correspond to the nodes reached through a prescription assigning the same action for all $I$ states;

\item children left with $I-1$ possible infostates: $A \cdot (A-1)^{1} \cdot(I-1)$. They correspond to the nodes reached through a prescription assigning any of the $A$ actions to the state corresponding to the card drawn in this subtree (defined by the chance outcomes) and to other $I-2$ states, and assign any of the remaining $A-1$ actions to the remaining state;

\item children left with $I-2$ possible infostates: $A \cdot(A-1)^{2} \cdot \binom{I-1}{2}$. They correspond to the nodes reached through a prescription assigning any of the $A$ actions to the state corresponding to the card drawn in this subtree (defined by the chance outcomes) and to other $I-3$ states, and assign any of the remaining $A-1$ actions to the 2 remaining states;

\item ...

\item children left with $1$ possible infostates: $A \cdot(A-1)^{C-1} \cdot \binom{C-1}{C-1}$. They correspond to the nodes reached through a prescription assigning any of the $A$ actions to the state corresponding to the card drawn in this subtree (defined by the chance outcomes) and assign any of the remaining $A-1$ actions to the $I-1$ remaining states.
\end{itemize}

We can generalize this pattern. Children left with $i$ possible private states out of available $I$: 
\begin{equation*}
n(i,I) = A \cdot (A-1)^{I-i} \cdot \binom{I-1}{I-i} \text{ for } i \in [1,I]
\end{equation*}
As a check: 
\begin{align*}
\sum_{i=1}^C n(i,I) & = \sum_{i=1}^I A \cdot (A-1)^{I-i} \cdot \binom{I-1}{I-i} = \\
& = A \sum_{j=0}^{I-1} \binom{I-1}{j} (A-1)^j 1^{I-1-j} =\\
& = A \cdot [(A-1)+1]^{I-1} = A^I
\end{align*}
which corresponds to the expected $A^I$ prescriptions available in the current node.

Then repartition of each level's nodes will depend on the number of nodes having a certain number of private states in the previous state, according to the repartition indicated by $n(i,I)$. In particular, each of the $b_{l-1}(c)$ nodes will generate $n(i,c)$ children with $i$ private states.

The correspondent succession is:
\begin{align*}
y_0(C) & = 
\left\{\begin{array}{ll}
0 & \text{if } c \not = C\\
C & \text{if } c = C
\end{array}\right.
\text{initial C chance outcomes}\\
y_l(c) & = \sum_{i=c}^C b_{l-1}(i) \cdot n(c,i)
\end{align*}
Note that we do not count auxiliary chance nodes, since in practical implementation they can be easily compacted with the previous coordinator nodes.

Therefore:
\begin{equation*}
\mathrm{tot}(C,A,H) = \sum_{l=0}^H \sum_{c=1}^C y_l(c)
\end{equation*}

\medskip
\textbf{Folding Representation}. In this representation, we have no initial chance sampling, and each coordinator node presents a chance node per prescription, each with a variable number of children depending on the number of unique actions.

To compute the total number of nodes per level, we can acknowledge that each coordinator node with $c$ private states corresponds to $c$ nodes (all with $c$ private states) in the pruning representation.
Therefore, at each level we have a number of coordinator node $z_l(c) = y_l(c)/c$

In this case, chance nodes have to be considered in the total nodes computed, since they cannot be easily reduced. In particular, each coordinator node has associated a chance node per prescription action available.

Therefore:
\begin{equation*}
z_l(c) = y_l(c) / c
\end{equation*}
\begin{equation*}
\mathrm{tot}(C,A,H) = \sum_{l=0}^H \sum_{c=1}^C y_l(c)/c \cdot (A^c +1)
\end{equation*}

Moreover, such an analysis can be extended also to the case of 2 initial levels of chance nodes, extracting one out of $C$ private states for P1 and P2 respectively. In this case, basic and pruning representation have a different starting condition, with $x_0(C) = y_0(C) = C^2$, while the folding representation has no changes.

\section{Experimental settings} \label{app:exp}
\subsection{Poker instances}
We refer to the three-player generalizations of Kuhn and Leduc poker proposed by \cite{Farina2018ExAC}.

Like all poker games, at the start of the game each player antes one to the pot, and receives a private card. Then players play sequentially in turn. Each player may check by adding to the pot the difference between the higher bet made by other players and their current bet (i.e. by matching the maximum bet made by others). Each player may fold whenever a check requires putting more money into the pot and the player instead decides to withdraw. Each player may raise whenever the maximum number of raises allowed by the game is not reached, by adding to the pot the amount required by a check plus an extra amount called raise amount. A betting round ends when all non-folded players except the last raising player have checked.

In \textbf{Kuhn poker}, there are three players and k possible ranks with k different ranks. The maximum number of raises is one, and the raising amount is 1. At the end of the first round, the showdown happens. The player having the highest card takes all the pot as payoff.

In \textbf{Leduc poker}, there are three players, k possible ranks having 3 cards in the deck each, and 1 or 2 raises. The raise amount is 2 for the first raise and 4 for the second raise. At the end of the first round, a public card is shown, and a new round of betting starts from the same player starting in the first round. In the end, the showdown happens. Winning players are having a private card matching the rank of the public card. If no player forms a pair, then the winning player is the one with the card with the highest rank. In the case of multiple winners, the pot is split equally.

\subsection{Implementation and execution details}
We implemented the folded representation of both Kuhn and Leduc taking advantage of the OpenSpiel \cite{lanctot2019openspiel} framework. The framework allowed us to specify the game as an evolving state object and provided the standard resolution algorithms for the computation of a Nash Equilibrium in the converted game.

The implementation is in Python3.8 and the experiments have been performed on a machine running Ubuntu 16.04 with a 2x Intel Xeon E5-4610 v2 @ 2.3GH CPU. The implementation is single-threaded.

\section{Converted game size} \label{app:sizes}
In the following, we present additional details on the sizes of the game instances obtained by converting Kuhn and Leduc poker games using \algoname{PublicTeamConversion}.

\begin{table}[H]
\centering
\makebox[0pt]{
\begin{tabular}{c|ccccccccc}
\toprule
Number of ranks						& 3		& 3		& 3		& 4		& 4		& 4		& 5		& 5		&		5\\
Adversary position					& 0		& 1		& 2		& 0		& 1		& 2		& 0		& 1		&	2\\
\midrule
Coordinator nodes					& 222   & 291	& 591	& 1560	& 2220	& 7412	& 8890	& 13025	& 66465\\
Adversary nodes						& 219	& 372	& 288	& 1996	& 5416	& 2656	& 12425	& 54040	& 16560\\
Terminal nodes						& 1320	& 1704	& 2436	& 16584	& 24536	& 51800	& 144740& 235660& 760520\\
Chance nodes						& 1129	& 1405	& 2461	& 10913	& 14641	& 40977	& 85521	& 119001& 514681\\
Chances with one child only			& 936	& 1188	& 2184	& 5680	& 7944	& 25400	& 29840	& 43360	& 218940\\
Total number of nodes				& 2890	& 3772	& 5776	& 31053	& 46813	& 102845& 251576& 421726& 1358226\\
Coordinator information sets		& 86	& 113	& 155	& 392	& 556	& 856	& 1738	& 2543	& 4093\\
Adversary information sets			& 12	& 12	& 12	& 16	& 16	& 16	& 20	& 20	& 20\\
Time taken for a full traversal		& 2.0s	& 2.3s	& 3.36s	& 14.7s	& 18.1s	& 37.2s	& 68.6s	& 125s	& 447s\\
\bottomrule
\end{tabular}
}
\vspace{3pt}
\caption{Converted Kuhn game characteristics for varying parameters.}
\label{tab:tabKuhn}
\end{table}

\begin{table}[H]
\centering
\makebox[0pt]{
\begin{tabular}{c|cccccc}
\toprule
Number of ranks						& 3		& 3			& 3			& 4		& 4		& 4	\\
Number of raises					& 1		& 1			& 1			& 2		& 2		& 2 \\
Adversary position					& 0		& 1			& 2			& 0		& 1		& 2	\\
\midrule
Coordinator nodes					& 84243 & 117126	& 232950	& 57138		& 66268		& 76384 \\
Adversary nodes						& 60543	& 98034		& 134196	& 32790		& 38622		& 46758 \\
Terminal nodes						& 354999& 476187	& 775233	& 163580	& 185994	& 213098\\
Chance nodes						& 284200& 378928	& 694132	& 160395	& 184065	& 211437\\
Chances with one child only			& 181020& 250908	& 494544	& 137044	& 159202	& 184738\\
Total number of nodes				& 783985& 1070275	& 1836511	& 413903	& 474949	& 547677\\
Coordinator information sets		& 7184	& 7232		& 7316		& 5624		& 5632		& 5650	\\
Adversary information sets			& 228	& 228		& 228		& 630		& 630		& 630	\\
Time taken for a full traversal		& 332s	& 322s		& 686s		& 220s		& 255s		& 183s	\\
\bottomrule
\end{tabular}
}
\vspace{3pt}
\caption{Converted Leduc game characteristics for varying parameters.}
\label{tab:tabLeduc}
\end{table}

\end{document}